\documentclass[letterpaper, 10 pt, conference]{ieeeconf}  

\IEEEoverridecommandlockouts                              
\usepackage{amsmath, amssymb, amsthm}
\usepackage{multirow}
\usepackage{booktabs}
\newtheorem{theorem}{Theorem}
\newtheorem{assumption}{Assumption}

\newtheorem{proposition}{Proposition}
\newtheorem{definition}{Definition}
\usepackage{graphicx}
\usepackage{caption}
\usepackage{url}

\DeclareMathOperator*{\minimize}{minimize}

\title{\LARGE \bf
An Adaptive Multi-Parameter ADMM Algorithm for Embedded MPC
}

\author{Alberto Zaupa\textsuperscript{\ddag} and Mikael Johansson\textsuperscript{\ddag}
\thanks{\textsuperscript{\ddag} The authors are with 
the Department of Decision and Control Systems, KTH Royal Institute of Technology, Stockholm, Sweden. 
        E-Mail: {\tt\small $\{$zaupa, mikaelj$\}$@kth.se}
        }
}

\begin{document}

\maketitle
\thispagestyle{empty}
\pagestyle{empty}

\begin{abstract}
We introduce an adaptive multi-parameter variant of ADMM and prove that it exhibits local superlinear convergence once the set of active constraints has been identified. In simulations, the proposed algorithm consistently outperforms OSQP, a standard ADMM solver, in terms of iteration count. We then implement an MPC solver based on our method and compare its runtime against a broader selection of state-of-the-art algorithms. Evaluations on challenging benchmark problems reveal that our approach delivers competitive performance both in terms of average and worst-case solve times, without being limited to coarse tolerances, as is typically the case for standard ADMM implementations and first-order methods.
\end{abstract}

\section{Introduction}

The Alternating Direction Method of Multipliers (ADMM) has become a popular optimization technique across a wide range of engineering fields from statistical learning and signal processing to control theory~\cite{boyd_admm,sra,Stellato2020}. One reason for its popularity is that it can exploit problem structure, e.g. handling non-smooth objectives or decomposing large problems into parallelizable subproblems. 
Another reason for its popularity is that it has global convergence guarantees for all positive values of its penalty parameter. In practice, ADMM is typically quick to find solutions of moderate accuracy, while convergence to high accuracy is much slower. It can also make effective use of warm-starts. These features make ADMM particularly attractive for Model Predictive Control (MPC), where optimization problems at successive sampling instances are very similar and previous solutions can be used to craft good initial guesses. Furthermore, since the optimizer then operates inside a feedback loop, solutions of moderate accuracy are often acceptable. In embedded settings, however, MPC is a latency-constrained application for which solvers must deliver sufficiently accurate solutions within a strict execution time budget. Improving the reliability and convergence speed of ADMM is therefore essential for extending its use to more demanding MPC scenarios. 

In this paper, we propose an adaptive multi-parameter ADMM scheme and prove its local superlinear convergence. 
The analysis here is restricted to box-constrained problems, but numerical evidence suggests that a similar result holds also in the general setting.
To the best of our knowledge, this is the first ADMM method for which superlinear convergence has been established. 
Existing results for ADMM-type methods describe sublinear active-set identification followed by linear local convergence; the superlinear rate established here represents a strict improvement, implying a much more rapid local reduction of residuals and optimality errors and making high-accuracy solutions attainable in substantially fewer iterations than with classical ADMM.

We further implement the proposed scheme in a high-performance solver and show that it matches or outperforms state-of-the-art solvers such as OSQP, qpOASES, HPIPM, and QPALM \cite{Stellato2020,Ferreau2014,hpipm,qpalm}. 
The solver is competitive in both latency and throughput, as measured by worst-case and average solve times, respectively. This is a relatively uncommon combination for MPC solvers, which typically favor one of these performance metrics at the expense of the other, with notable exceptions such as QPALM.

\begin{figure}[t] 
    \centering 
    \includegraphics[width=0.875\linewidth]{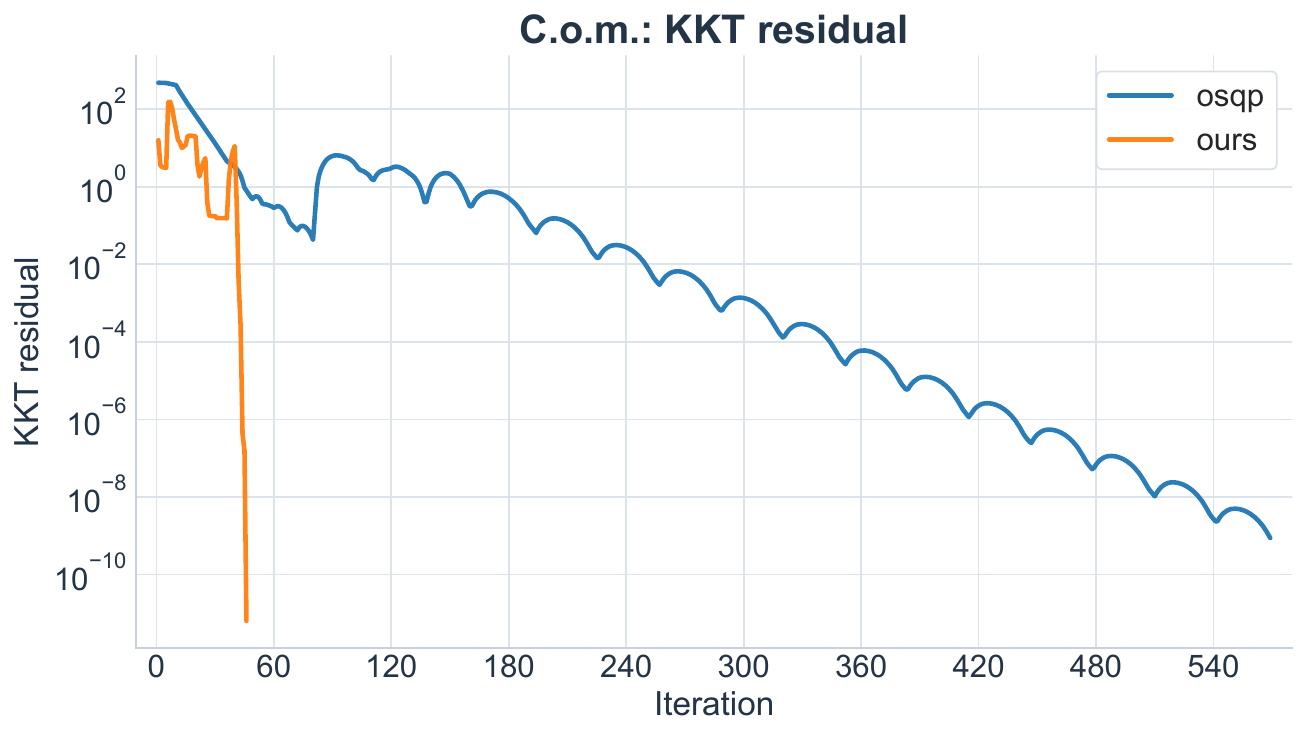} 
    \caption{Evolution of the KKT residual norm for OSQP and the proposed ADMM algorithm on one of the problems listed in Table \ref{tab:osqp_comp}. After our algorithm identifies the active set, its convergence speed increases dramatically.}
\label{fig:osqp_comp}
\end{figure}

This paper is organized as follows. Section~\ref{sec:background} reviews the necessary background on MPC and ADMM. Section~\ref{sec:amp-admm} then introduces our adaptive multi-parameter penalty scheme and proves the local superlinear convergence rate. In Section~\ref{sec:implementation}, we compare the proposed MPC solver with state of the art solutions on several benchmarks. Finally, in Section~\ref{sec:conclusions} we conclude the paper and identify directions for future work.

\section{Background} \label{sec:background}

\subsection{Model Predictive Control}
In linear MPC~\cite{borrelli2017predictive}, at each sampling time we solve an optimization problem of the form:
\begin{equation}
\label{eq:mpc}
\begin{aligned}
    \minimize_{\{x_t\}\{u_t\}} &\quad \sum_{t=0}^{N-1} \frac{1}{2} x_t^{\top} Q x_t + q_t^{\top} x_t + \frac{1}{2} u_t^{\top}R u_t + r_t^{\top} u_t\\
    &\quad + \frac{1}{2} x_N^{\top} Q_f x_N + q_f^{\top}x_N \\
    \text{subject to} &\quad x_{t+1} = A x_t + B u_t \quad t=0,1,\dots,N-1 \\
    &\quad c_x \leq C_x x_t \leq d_x \quad t=1,2,\dots, N\\ 
    &\quad c_u \leq C_u u_t \leq d_u \quad t=0,1,\dots,N-1 \\
    &\quad x_0 \text{ given.}
\end{aligned}
\end{equation}
Here $x_0$ is the state of the plant, $A$ and $B$ are the system matrices, $N \in \mathbb{N}$ is the horizon length, $Q=Q^{\top} \succeq 0$ and $R=R^{\top}\succ0$ are the stage cost matrices and $Q_f = Q_f^{\top}\succ0$ is the terminal cost matrix. The terms $q_f$, $\{q_t\}$ and $\{r_t\}$ arise when tracking a non-zero reference trajectory, and the stage-wise inequality constraints enforce operational and physical limits of the plant and the actuators.

In this paper, we consider two standard solution methods for~\eqref{eq:mpc}.
The first approach condenses the problem by eliminating the state variables to arrive at a QP in the form
\begin{equation}
\label{eq:dense_qp}
\begin{aligned}
    \minimize_w &\quad \frac{1}{2}w^{\top} H w + g^{\top} w \\
    \text{subject to} &\quad a \leq Gw \leq b
\end{aligned}
\end{equation}
where the Hessian $H = H^{\top} \succ 0$ is dense. The decision vector $w$ is formed by stacking 
the controls at all stages of the prediction horizon into a single vector. In the MPC setting, condensed QPs are often solved using qpOASES~\cite{Ferreau2014}.

The second approach retains the states as optimization variables, as is done in solvers like HPIPM~\cite{hpipm}. The resulting QPs preserve the sparse and structured optimal control problem (OCP) formulation. This structure is important because it can be exploited by tailored linear-algebra routines, in particular Riccati-based methods, leading to substantially more efficient QP solves than would be obtained by treating the problem as an unstructured sparse program~\cite{riccati_frison}. The reformulated problem is
\begin{equation}
\label{eq:sparse_qp}
\begin{aligned}
    \minimize_{w} &\quad \frac{1}{2} w^{\top} H w + g^{\top} w + \mathcal{I}_{\mathcal{X}}(w) \\
    \text{subject to} &\quad a \leq G w \leq b
\end{aligned}
\end{equation}
where now $w$ is obtained by stacking both $\{u_t\}$ and $\{x_t\}$ over the horizon, and therefore $H$ and $G$ are block diagonal. The set $\mathcal{X}$ is the affine space where the dynamical constraints are satisfied, and $\mathcal{I}_{\mathcal{X}}$ is the indicator function of $\mathcal{X}$.

\subsection{ADMM}
ADMM~\cite{boyd_admm} is a standard algorithm for solving convex optimization problems of the form:
\begin{equation}
\label{eq:admm_problem}
\begin{aligned}
    \minimize_{w,z} &\quad f(w) + h(z) \\
    \text{subject to} &\quad M w + N z = c
\end{aligned}
\end{equation}
Note that by the appropriate choice of $f$, we can rewrite both~\eqref{eq:dense_qp} and~\eqref{eq:sparse_qp} in the form:
\begin{equation}
\label{eq:admm_mpc}
\begin{aligned}
    \minimize_{w,z} &\quad f(w) + \mathcal{I}_{[a,b]}(z) \\
    \text{subject to} &\quad Gw = z
\end{aligned}
\end{equation}
which is a particular instance of~\eqref{eq:admm_problem}. Solving Problem~\eqref{eq:admm_mpc} through ADMM involves iterating the following updates:
\begin{equation}
\label{eq:admm_iters}
\begin{aligned}
w^{(k+1)} &= \arg \min_w f(w) + \frac{1}{2} \| Gw - z^{(k)} + \rho^{-1} y^{(k)}\|_{\rho}^2 \\
z^{(k+1)} &= \Pi_{[a,b]} \left(Gw^{(k+1)} + \rho^{-1} y^{(k)} \right) \\
y^{(k+1)} &= y^{(k)} + \rho\, (G w^{(k+1)} - z^{(k+1)})
\end{aligned}
\end{equation}
where $\rho$ is a positive definite \emph{diagonal matrix}. This is more general than the standard scalar-penalty form, and allows different penalty parameters for different constraint components. 
Here superscripts between parenthesis like $w^{(k)}$ denote iteration count. The primal and dual KKT residuals for ADMM take the following form~\cite{boyd_admm}:
\begin{equation}
\label{eq:residuals}
\begin{aligned}
r_p^{(k)} &= G w^{(k)} - z^{(k)} \\
r_d^{(k+1)} &= G^{\top} \rho\, (z^{(k+1)} - z^{(k)})
\end{aligned}
\end{equation}
and we run the algorithm until $\|r_p \|_{\infty} < \varepsilon$ and $\|r_d \|_{\infty} < \varepsilon$. Under standard assumptions on $f$, the iterates generated by \eqref{eq:admm_iters} converge to a solution of~\eqref{eq:admm_mpc} for any value of $\rho$.

Common choices for the penalty matrix are $\rho = \varrho I$  where $\varrho \in \mathbb{R}^+$ (the standard single $\rho$), or $\rho = \varrho\, P^2$, where $P$ is a diagonal matrix that acts as a constraint preconditioner~\cite{Stellato2020}.

\subsection{Adaptive ADMM}
For fixed-parameter ADMM, the convergence speed can be very sensitive to the choice of $\rho$, and for some specific classes of QPs one can compute the optimal combination of preconditioner $P$ and scalar penalty $\varrho$ that leads to the fastest linear convergence rate~\cite{optimal_rho}. However, in practice this is often either too expensive or just not possible. For this reason, ADMM solvers typically rely on adaptive schemes that adjust the value of a scalar penalty $\varrho$ during the execution of the algorithm. The simplest adaptive scheme is the so called residual balancing heuristic~\cite{boyd_admm}:
\begin{equation}
\label{eq:rb}
\begin{aligned}
    \varrho^{(k+1)} = \begin{cases}
        \delta \varrho^{(k)} \quad&\text{if}\quad \|r_p^{(k+1)}\| > \mu \|r_d^{(k+1)}\|\\ 
        \frac{1}{\delta} \varrho^{(k)} \quad&\text{if}\quad \|r_d^{(k+1)}\| > \mu \|r_p^{(k+1)}\|\\ 
        \varrho^{(k)} \quad &\text{otherwise}
    \end{cases}
\end{aligned}
\end{equation}
with $\delta,\mu > 1$. An alternative variant is:
\begin{equation}
\label{eq:rb_sqrt}
    \varrho^{(k+1)} = \varrho^{(k)} \sqrt{\frac{\|r_p^{(k+1)}\|}{\|r_d^{(k+1)}\|}}
\end{equation} 

Note that in the varying-$\rho$ case, the expression for the dual residual is $r_d^{(k+1)} = G^{\top}\rho^{(k)}(z^{(k+1)} - z^{(k)})$.
There also exist more complex versions of residual balancing, that, for example, are designed to be independent of problem scaling~\cite{residual_balancing, Stellato2020}. 
Recently,~\cite{multiblock_admm} proposed a penalty adaptation scheme for problems where the constraint matrix $G$ is block diagonal. The approach associates a different penalty parameter to each block of $G$ and updates them independently using the following SRA heuristic~\cite{sra}:
\begin{equation}
\label{eq:sra}
    \rho_{\mathcal{I}_i}^{(k+1)} = \frac{\|\left( y^{(k+1)} - y^{(k)} \right)_{\mathcal{I}_i} \|}{\|\left( z^{(k+1)} - z^{(k)} \right)_{\mathcal{I}_i} \|}
\end{equation}
where $\mathcal{I}_i$ spans the set of block indices of $G$.

This \textit{multi-block} scheme is well suited for MPC  when we solve problems through the sparse formulation~\eqref{eq:sparse_qp}. Indeed, in our own experiments we find that the scheme often improves convergence speed over single-parameter ADMM, especially for box-constrained problems. However, improvements are less significant on non-box-constrained problems. Moreover, this method is in general not applicable when we solve the condensed problem~\eqref{eq:dense_qp}.

\section{Adaptive multi-parameter ADMM} \label{sec:amp-admm}

In this section we introduce the proposed multi-parameter update for the penalty matrix $\rho$. For box-constrained problems,
our method can be interpreted as a per-constraint version of the square-root residual-balancing rule~\eqref{eq:rb_sqrt}:
\begin{equation}
\label{eq:multiparameter_rb}
\begin{aligned}
\rho_i^{(k+1)} = \rho_i^{(k)} \sqrt{\frac{|r_p^{(k+1)}|_i}{|r_d^{(k+1)}|_i}} .
\end{aligned}
\end{equation}
In our numerical experiments, this simple rule proved remarkably effective: it often accelerated convergence substantially while remaining very stable. Compared with the ``dead-band'' version of residual balancing in~\eqref{eq:rb}, the square-root variant generally led to faster convergence while also introducing fewer tuning parameters.

Note that \eqref{eq:multiparameter_rb} can be rewritten as
\begin{align*}
    \rho_i^{(k+1)} = \sqrt{\frac{|y^{(k+1)}-y^{(k)}|_i}{|z^{(k+1)}-z^{(k)}|_i}},
\end{align*}
so that its right-hand side is precisely the square root of the right-hand side of the SRA heuristic~\eqref{eq:sra}. This relation is worth emphasizing. Although in~\cite{multiblock_admm} Lozenski \emph{et al.} discuss SRA primarily in the multi-block setting, the quantities appearing in the update are naturally $m$-dimensional, and the heuristic can therefore also be interpreted on a per-constraint basis. Our update may thus be viewed as a more conservative variant of SRA, obtained by taking the square root of its update factor. 
In our numerical experiments we found SRA to be generally more aggressive than~\eqref{eq:multiparameter_rb}, sometimes leading to slightly faster convergence, while often diverging.

For general $G$, the primal and dual residuals have different dimensions, so we cannot simply apply (\ref{eq:multiparameter_rb}). Nevertheless, $ \rho^{(k)}(z^{(k+1)}-z^{(k)})$ is still a per-constraint pre-cursor of the dual residual and it measures how much each individual constraint copy $z_i$ is moving. Hence, we propose to base our update on
\begin{equation}
\label{eq:update_rule_naive}
\rho_i^{(k+1)} = \rho_i^{(k)} \sqrt{\frac{|G w^{(k+1)} - z^{(k+1)}|_i}{|\rho^{(k)}\left(z^{(k+1)} - z^{(k)}\right)|_i}}.
\end{equation}

For numerical stability and reliability, we add a number of safeguards to this update. The first one ensures that the entries of $\rho^{(k)}$ remain strictly positive and finite by replacing the basic update by 
\begin{equation}
\begin{aligned}
\phi_i^{(k+1)} &= \frac{\max\{ \eta, \; |Gw^{(k+1)} - z^{(k+1)}|_i \}}{\max\{ \eta, \; |\rho^{(k)} (z^{(k+1)} - z^{(k)})|_i \}}    \\
\hat{\rho}_i^{(k+1)} &= \rho^{(k)}_i \sqrt{\phi_i^{(k+1)}} 
\end{aligned}
\end{equation}
for some small parameter $\eta>0$. The next one modifies $\hat{\rho}$ into a corrected intermediate quantity $\tilde{\rho}$
that reflects the active/inactive structure of constraints. Specifically, we define
\begin{equation}
\label{eq:rho_tilde}
\begin{aligned}
    \tilde{\rho}_i^{(k+1)} = \begin{cases}
        \min\{\hat{\rho}_i^{(k+1)}, \; \mu_1 \rho_i^{(k)}\} \; &\text{if}\; |y^{(k+1)}|_i \leq \epsilon \\
        \max\{\hat{\rho}_i^{(k+1)}, \; \mu_2 \rho_i^{(k)}\} \; &\text{otherwise}
    \end{cases}
\end{aligned}
\end{equation}
with $\mu_1 \in (0,1)$ and $\mu_2 > 1$. The rationale for this modification is that if the dual variable is essentially zero, the corresponding coordinate behaves as an inactive constraint and we enforce a mild multiplicative decrease. Otherwise, if the dual variable is non-zero the corresponding constraint is currently active, and therefore we enforce a mild multiplicative increase of its penalty. 
Thus, $\tilde{\rho}$ can be viewed as an active-set-aware correction of $\hat{\rho}$ rather than a separate adaptation principle. Finally, we still do not apply $\tilde{\rho}$. Large instantaneous changes in $\rho$ can deteriorate the conditioning of the $w$-update and destabilize the method. For this reason, the final penalty update is obtained by first clipping $\tilde{\rho}$ to a prescribed interval $[\underline{\rho}, \overline{\rho}]$, and then moving only part of the way toward the clipped value in the log-domain. Concretely, we set
\begin{equation}
\label{eq:update_rule}
    \rho_i^{(k+1)} =\left( \rho^{(k)}_i \right)^{1-\alpha} \Pi_{[\underline{\rho}, \overline{\rho}]} \left(\tilde{\rho}_i^{(k+1)}\right)^{\alpha} 
\end{equation}
where $\alpha\in (0,1]$ and $\Pi_{[\underline{\rho}, \overline{\rho}]}$ denotes projection onto the interval 
$[\underline{\rho}, \overline{\rho}]$. This is a log-domain under-relaxation step: when $\alpha=1$, we accept the clipped candidate in full, while for smaller values of $\alpha$, we interpolate geometrically between the previous penalty and the new candidate. The clipping guarantees bounded $\log\rho^{(k)}$ while the under-relaxation reduces oscillations and improves stability.

In the remainder of this paper we will fix $\alpha = \frac{1}{2}$, $\underline{\rho} = 10^{-6}$, $\overline{\rho} = 10^6$, $\mu_1 = 0.99$, $\mu_2 = 1.01$ and $\eta = \epsilon = 10^{-15}$. Furthermore, we apply the update rule in~\eqref{eq:update_rule} every $K = 5$ iterations.

\subsection{Comparison with OSQP}

To highlight the main features of the proposed update rule, we compare its convergence speed with OSQP. This is a natural baseline, since OSQP employs an adaptive single-parameter $
\varrho$ update closely related to residual balancing. Table~\ref{tab:osqp_comp} reports the number of iterations to convergence on five benchmark MPC problems drawn from the literature: the Quadrotor problem, a standard quadcopter control problem with linearized dynamics; the Atlas and Quadruped benchmarks from~\cite{reluqp}; and the Chain of masses (C.o.m.) and Aircraft problems from~\cite{mpc_benchmarks}. For each benchmark we solve a single problem instance, without warm start. OSQP is run with default settings, except that the frequency of $\varrho$ updates is matched to our method, i.e., $K=5$. Out of the candidates $\{5, 10, 15, 20, 25, 50\}$, we found the update frequency given by $K=5$ to lead to the fastest convergence for OSQP.

Across all benchmarks, the proposed algorithm converges substantially faster than OSQP and, remarkably, its performance is nearly insensitive to the target tolerance $\varepsilon$. This behavior is unusual for ADMM and related first-order methods, which typically slow down considerably at tight tolerances. Figure~\ref{fig:superlinear_convergence} and Table~\ref{tab:active_set_id} suggest that this behavior is tied to active-set identification: after some iteration $\bar K$, the signs of the dual variables cease to change and the active set at the solution is identified, after which convergence accelerates sharply. Active-set identification is well documented for standard ADMM with fixed $\rho$~\cite{as_identification}, where one can obtain a sublinear transient followed by fast linear convergence once the active set is identified. Our observations are closely related, but differ in two important respects: existing results do not cover adaptive $\rho$ updates, and the convergence regime we observe appears faster than linear, with each $\rho$ update further accelerating convergence. This motivates the theoretical analysis developed in the next section.

\begin{table}
\centering
\begin{tabular}{llccc}
\toprule
Problem & Solver & $\varepsilon=10^{-3}$ & $\varepsilon=10^{-6}$ & $\varepsilon=10^{-9}$ \\
\midrule
\multirow{2}{*}{Atlas} 
    & Ours & \textbf{18} & \textbf{21} & \textbf{22} \\
    & OSQP & 61 & 103 & 132 \\
\midrule
\multirow{2}{*}{Quadruped} 
    & Ours & \textbf{32} & \textbf{36} & \textbf{40} \\
    & OSQP & 156 & 226 & 387 \\
\midrule
\multirow{2}{*}{Quadrotor} 
    & Ours & \textbf{42} & \textbf{43} & \textbf{45} \\
    & OSQP & 227 & 443 & 659 \\
\midrule
\multirow{2}{*}{C.o.m.} 
    & Ours & \textbf{43} & \textbf{44} & \textbf{46} \\
    & OSQP & 285 & 439 & 569 \\
\midrule
\multirow{2}{*}{Aircraft} 
    & Ours & \textbf{42} & \textbf{43} & \textbf{44} \\
    & OSQP & 2965 & - & - \\
\midrule
\bottomrule
\end{tabular}
\caption{We compare iterations until convergence on 5 different problems for our solver and OSQP, and over different tolerances for the termination criteria. On the Aircraft problem, OSQP cannot meet the termination criteria for $\varepsilon \in \{10^{-6}, 10^{-9}\}$.}
\label{tab:osqp_comp}
\end{table}

\begin{figure}[t] 
    \centering 
    \includegraphics[width=\linewidth]{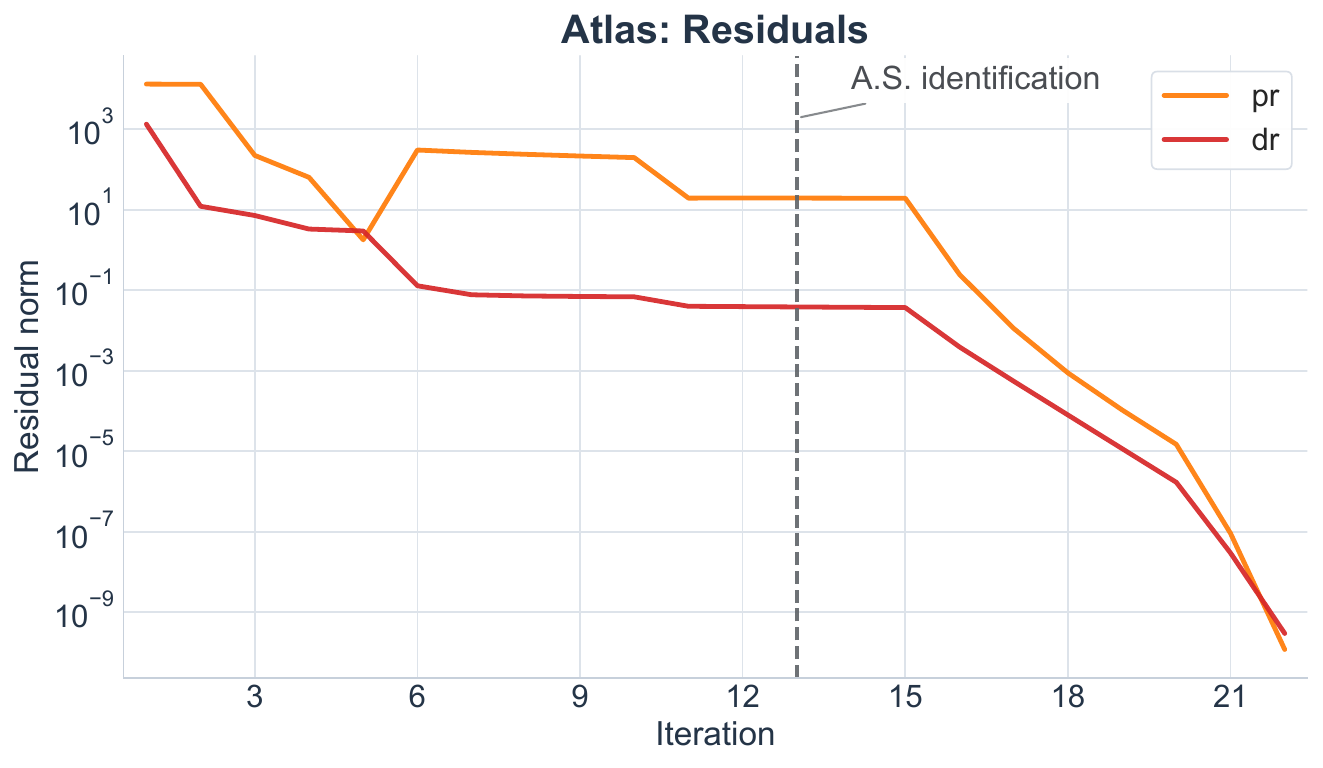} 
    \caption{KKT residuals for the proposed ADMM algorithm on the Atlas problem.}
\label{fig:superlinear_convergence}
\end{figure}

\begin{table}
\centering
\begin{tabular}{lccc}
\toprule
Problem & $\bar{K}$  & Convergence to $\varepsilon=10^{-9}$ & $\|r^{(\bar{K})}\|_{\infty}$ \\
\midrule
Atlas & 13 & 22 & 19.55 \\
\midrule
Quadruped & 28 & 40 & 0.21 \\
\midrule
Quadrotor & 35 & 45 & 1.82 \\
\midrule
C.o.m & 42 & 46 & $3.6\cdot10^{-3}$ \\
\midrule
Aircraft & 38 & 44 & 2.06\\
\bottomrule
\end{tabular}
\caption{After the Active Set is identified at iteration $\bar{K}$, only a couple of $\rho$ updates (happening every 5 iterations) are necessary for the $\infty$-norm of KKT residuals to fall below $\varepsilon=10^{-9}$. Note that before active set identification, we are several orders of magnitude away from convergence.}
\label{tab:active_set_id}
\end{table}

\subsection{Theoretical Analysis}
\label{sec:theoretical_analysis}

Due to space limitations, we do not present a complete analysis of the algorithm here. Instead, we focus on its most novel convergence behavior, namely local superlinear convergence, and we carry out the analysis for box-constrained problems with $G=I$. Moreover, for simplicity of exposition we assume that we update $\rho^{(k)}$ in every iteration. 

Our goal in this subsection is not to prove active-set identification, but to characterize the asymptotic regime after identification has taken place. Accordingly, we assume that after some finite iteration $\bar{K}$, the active set of the ADMM iterates coincides with the active set at the solution. 
In the following we formalize this notion in terms of the behavior of $y^{(k)}$ and $z^{(k)}$. We begin by defining the concept of \emph{active} and \emph{inactive} constraints at a solution.

\begin{definition}
Given a primal-dual solution of Problem~\eqref{eq:admm_mpc} $(w^{\star}, z^{\star}, y^{\star})$, we say that constraint $i$ is active at the triple if $y_i^{\star} \neq 0$. Otherwise we say that $i$ is inactive.
\end{definition}
\begin{assumption}[Post-identification regime]
\label{assumption:1}
Let $(w^{\star}, z^{\star}, y^{\star})$ be a primal-dual solution of Problem~\eqref{eq:admm_mpc}, and let $\mathbb{A}$ and $\mathbb{I}$ denote the active and inactive coordinates at the solution. We assume that there exists $\bar{K} \in \mathbb{N}$ such that for all $k \geq \bar{K}$, the ADMM iterates generated by~\eqref{eq:admm_iters} satisfy
\begin{equation}
\label{eq:assumtion_1}
\text{sign}(y_i^{(k)}) = \text{sign}(y_i^{\star}) \quad \forall i =1,2,\dots, m
\end{equation}
where we define $\text{sign}(0) = 0$.
\end{assumption}

As the following proposition shows, Assumption~\ref{assumption:1} implies that in the post-identification regime, the active set of the iterates is fixed and agrees with the active set at the solution.

\begin{proposition}
\label{prop:assumption_1}
Let Assumption~\ref{assumption:1} hold. Then for all $k > \bar{K}$:
\begin{equation}
\begin{aligned}
    z_{\mathbb{A}}^{(k)} &= z_{\mathbb{A}}^{\star} \\
    z_{\mathbb{I}}^{(k)} &= w_{\mathbb{I}}^{(k)}
\end{aligned}
\end{equation}
\end{proposition}
\begin{proof}
Let $i \in \mathbb{A}$, meaning that $y_i^{\star} \neq  0$. Suppose in particular that $y_i^{\star} > 0$. According to Assumption~\ref{assumption:1}, $\forall k \ge \bar{K}$ $y_i^{(k+1)} >0$. By definition of the dual update, this in turn is equivalent to:
\begin{equation}
\label{eq:z_active_is_z_star}
\begin{aligned}
    y_i^{(k)} + \rho_i^{(k)} \left(w_i^{(k+1)} - z_i^{(k+1)} \right) &> 0 \iff \\
    \frac{y_i^{(k)}}{\rho_i^{(k)}} + w_i^{(k+1)} &> z_i^{(k+1)} \,.
\end{aligned}
\end{equation}
Given the definition of the $z$ update in~\eqref{eq:admm_iters}, Equation~\eqref{eq:z_active_is_z_star} implies $z_i^{(k)} = b_i$, $\forall k > \bar{K}$. However, because $y_i^{\star} > 0$, we must have $z_i^{\star} = b_i$, and thus we showed that $\forall k > \bar{K}$ $z_i^{(k)} = z_i^{\star}$. Similar reasoning applies to the case $y_i^{\star} < 0$.

Now let $i \in \mathbb{I}$, meaning that $y_i^{\star} = 0$. Then by Assumption~\ref{assumption:1}, $\forall k \geq \bar{K}$ $y_i^{(k)} = 0$. This in turn implies:
\begin{equation}
\label{eq:z_inactive_is_w}
\begin{aligned}
0 &= y_i^{(k+1)} \\
&= y_i^{(k)} + \rho_i^{(k)} \left(w_i^{(k+1)} - z_i^{(k+1)}\right) \\
&= \rho_i^{(k)} \left(w_i^{(k+1)} - z_i^{(k+1)}\right)
\end{aligned}
\end{equation}
and therefore we get $w_i^{(k)} = z_i^{(k)}$ $\forall k \geq \bar{K}$.
\end{proof}

The post-identification regime is also characterized by a specific trend in the $\rho^{(k)}$-dynamics induced by our update rule: penalty parameters associated to active constraints rapidly increase, while the remaining ones decay to zero.
\begin{proposition}
\label{prop:rho_dynamics}
    The proposed update rule~\eqref{eq:update_rule} with $\underline{\rho} = 0$, $\overline{\rho} = +\infty$ and $\epsilon=0$ is such that when Assumption~\ref{assumption:1} holds, $\rho$ updates satisfy:
    \begin{equation}
    \label{eq:assumption_2}
    \begin{aligned}
        \rho_i^{(k+1)} &\ge \beta\, \rho_i^{(k)} \quad \forall i \in \mathbb{A} \\
        \rho_i^{(k+1)} &\le \gamma\, \rho_i^{(k)} \quad \forall i \in \mathbb{I}
    \end{aligned}
    \end{equation}
for some $\beta > 1$ and $\gamma \in (0, 1)$.
\end{proposition}
\begin{proof}
Let $i \in \mathbb{A}$, and suppose that constraint $i$ is active at the upper bound. Then Assumption~\ref{assumption:1} implies that $y^{(k)}_i > 0$ $\forall k \geq \bar{K}$. Thus combining Equations~\eqref{eq:rho_tilde},~\eqref{eq:update_rule} and our assumptions on $\underline{\rho}$, $\overline{\rho}$, $\epsilon$:
\begin{equation}
    \rho_i^{(k+1)} \geq \beta \,\rho_i^{(k)}
\end{equation}
where $\beta = \mu_2^\alpha > 1$. Similar reasoning applies to the case where $i$ is active at the lower bound.

Now let $i \in \mathbb{I}$. Then Assumption~\ref{assumption:1} implies that $y^{(k)}_i = 0$ $\forall k \geq \bar{K}$, and therefore:
\begin{equation}
    \rho_i^{(k+1)} \leq \gamma \, \rho_i^{(k)}
\end{equation}
where $\gamma = \mu_1^\alpha \in (0,1)$.
\end{proof}

To assume $\epsilon = 0$ is not a fundamental restriction, since $\epsilon$ is only introduced to handle finite-precision arithmetic. At the end of this section we address the assumptions on $\underline{\rho}$ and $\overline{\rho}$, and how they relate to our practical  implementation. 

The behavior of $\rho^{(k)}$ when the active-set has been identified is the key mechanism that drives superlinear convergence, as we show in the following theorem, our main result.
\begin{theorem}
\label{thm:superlinear_convergence}
Let $(w^{\star}, z^{\star}, y^{\star})$ be a primal-dual solution of Problem~\eqref{eq:admm_mpc} with $G = I$, and denote by $\mathbb{A}$ and $\mathbb{I}$ the set of active and inactive constraints at the triple. Assume that after iteration $\bar{K}$, the sequence $\{(w^{(k)}, z^{(k)}, y^{(k)})\}$ generated by~\eqref{eq:admm_iters} satisfies Assumption~\ref{assumption:1}, and that $\{\rho^{(k)}\}$ is driven by~\eqref{eq:update_rule} with $\underline{\rho} = 0$, $\overline{\rho} = + \infty$ and $\epsilon=0$. Moreover suppose that there exists a closed ball $\mathcal{N}$ centered at $w^{\star}$ such that $f$ is $\mathcal{C}^2$ in $\mathcal{N}$, $\nabla^2 f(w) \succ 0 \; \forall w \in \mathcal{N}$ and $\{w^{(k)}\}_{k\geq\bar{K}} \in \mathcal{N}$. 
Then there exists $C > 0$ such that:
\begin{equation}
    \|w^{(k+1)} - w^{\star} \| \leq C e^{-\Theta(k^2)} \|w^{(0)} - w^{\star}\| \;.
\end{equation}
\end{theorem}

\begin{proof}
Because we are after an asymptotic convergence rate, in the following we may assume without loss of generality $\bar{K} = 0$.

Notice that given our assumptions on $\mathcal{C}^2$ continuity of $f$ within $\mathcal{N}$, the $w$-update is characterized by the optimality condition:
\begin{equation}
\label{eq:w_optimality}
\nabla f(w^{(k+1)}) + y^{(k)} + \rho^{(k)} (w^{(k+1)} - z^{(k)}) = 0
\end{equation}
and slicing the equation above with respect to $\mathbb{A}$ and $\mathbb{I}$ we get:
\begin{equation}
\label{eq:w_opt_manipulation_1}
\nabla f(w^{(k+1)}) + \begin{bmatrix}
    y_{\mathbb{A}}^{(k)} \\
    0
\end{bmatrix} + \begin{bmatrix}
    \rho_{\mathbb{A}}^{(k)} (w_{\mathbb{A}}^{(k+1)} - z_{\mathbb{A}}^{\star}) \\
    \rho_{\mathbb{I}}^{(k)} (w_{\mathbb{I}}^{(k+1)} - w_{\mathbb{I}}^{(k)})
\end{bmatrix} = 0
\end{equation}
where we used $z^{(k)}_\mathbb{A} = z_\mathbb{A}^{\star}$. $w_\mathbb{I}^{(k)} = z_\mathbb{I}^{(k)}$ and $y_{\mathbb{I}}^{(k)} = 0$, as stated by Proposition~\ref{prop:assumption_1} and Assumption~\ref{assumption:1}. However, by definition of the dual update in~\eqref{eq:admm_iters} and because $z_\mathbb{A}^{(k)} = z_\mathbb{A}^{\star}$ $\forall k$,~\eqref{eq:w_optimality} also implies:
\begin{equation}
\label{eq:w_opt_manipulation_2}
\begin{aligned}
\nabla_{w_{\mathbb{A}}} f(w^{(k+1)}) + y_{\mathbb{A}}^{(k)} + \rho_{\mathbb{A}}^{(k)} (w_{\mathbb{A}}^{(k+1)} - z_{\mathbb{A}}^{(k+1)}) &= 0 \iff \\
\nabla_{w_{\mathbb{A}}} f(w^{(k+1)}) + y_{\mathbb{A}}^{(k+1)} &= 0
\end{aligned}
\end{equation}
where we denote by $\nabla_{w_{\mathbb{A}}}f$ the gradient of $f$ w.r.t. $w_{\mathbb{A}}$. Combining~\eqref{eq:w_opt_manipulation_1} and~\eqref{eq:w_opt_manipulation_2}, observing that $w_{\mathbb{A}}^{\star} = z_{\mathbb{A}}^{\star}$ and defining $e^{(k)} = w^{(k)} - w^{\star}$ we get:
\begin{equation}
\label{eq:gradients_errors}
\nabla f(w^{(k+1)}) + \rho^{(k)} e^{(k+1)} = \begin{bmatrix}
    \nabla_{w_{\mathbb{A}}}f(w^{(k)}) \\
    \rho_{\mathbb{I}}^{(k)} e^{(k)}_{\mathbb{I}}
\end{bmatrix}
\end{equation}
where we added $0 = -w_\mathbb{I}^{*} + w_\mathbb{I}^{*}$ to the $w_\mathbb{I}^{(k+1)} - w_\mathbb{I}^{(k)}$ term in~\eqref{eq:w_opt_manipulation_1}.
Now we rewrite the gradients in \eqref{eq:gradients_errors} in terms of the error sequence. The mean value theorem states:
\begin{equation}
\begin{aligned}
    \nabla f(w^{(k)}) &= \nabla f(w^{\star}) + \\
    &\left( \int_0^1 \nabla^2 f(\tau w^{\star} + (1 - \tau) w^{(k)}) d\tau \right) e^{(k)} \\
    &= \nabla f(w^{\star}) +  H^{(k)} e^{(k)}
\end{aligned}
\end{equation}
where given our assumptions on $\nabla^2f$ we have that there exist $0 < m < M$ such that, for all $k$:
\begin{equation}
\label{eq:H_bounds}
mI \preceq H^{(k)} \preceq M I
\end{equation}
At this point, we observe that $\nabla_{w_\mathbb{I}}f(w^{\star}) = 0$, which follows immediately from the stationarity condition $\nabla f(w^{\star}) + y^{\star} = 0$, combined with $y_\mathbb{I}^{\star} = 0$. Therefore, defining:
\begin{equation}
    H^{(k)} = \begin{bmatrix}
        H_{\mathbb{A}, \mathbb{A}}^{(k)} & H_{\mathbb{A},\mathbb{I}}^{(k)} \\
        H_{\mathbb{I}, \mathbb{A}}^{(k)} & H_{\mathbb{I}, \mathbb{I}}^{(k)}
    \end{bmatrix}
\end{equation}
we can rewrite Equation~\eqref{eq:gradients_errors} as follows:
\begin{equation}
\begin{aligned}
    \nabla f(w^{{\star}}) + (H^{(k+1)} + \rho^{(k)}) e^{(k+1)} = \\
    \begin{bmatrix} \nabla_{w_{\mathbb{A}}}f(w^{\star}) + 
    H^{(k)}_{\mathbb{A},\mathbb{A}}\, e_{\mathbb{A}}^{(k)} + H_{\mathbb{A}, \mathbb{I}}^{(k)} \,e_{\mathbb{I}}^{(k)} \\
    \rho_{\mathbb{I}}^{(k)} e_{\mathbb{I}}^{(k)}
\end{bmatrix}
\end{aligned}
\end{equation}
where the gradients cancel out, leading to:
\begin{equation}
\label{eq:error_recursion}
(H^{(k+1)} + \rho^{(k)}) e^{(k+1)} = \begin{bmatrix}
    H^{(k)}_{\mathbb{A},\mathbb{A}}\, e_{\mathbb{A}}^{(k)} + H_{\mathbb{A}, \mathbb{I}}^{(k)} \,e_{\mathbb{I}}^{(k)} \\
    \rho_{\mathbb{I}}^{(k)} e_{\mathbb{I}}^{(k)}
\end{bmatrix}
\end{equation}
which is a recursion in the error $e^{(k)}$. We will now manipulate Equation~\eqref{eq:error_recursion} in order to obtain a rate of convergence for $\|e^{(k)}\|$. We start by solving for $e_{\mathbb{I}}^{(k+1)}$:
\begin{equation}
\label{eq:solve_I}
\begin{aligned}
    e_{\mathbb{I}}^{(k+1)} = \left(P^{(k+1)}\right)^{-1}\left( -H_{\mathbb{I}, \mathbb{A}}^{(k+1)} e_{\mathbb{A}}^{(k+1)} + \rho_{\mathbb{I}}^{(k)} e_{\mathbb{I}}^{(k)} \right)
\end{aligned}
\end{equation}
where $P^{(k+1)} = H_{\mathbb{I}, \mathbb{I}}^{(k+1)} + \rho_{\mathbb{I}}^{(k)}$. We then substitute in the system for $e_{\mathbb{A}}^{(k+1)}$:
\begin{equation}
\begin{aligned}
\label{eq:solve_A}
    S^{(k+1)} e_{\mathbb{A}}^{k+1} = B^{(k+1)} e^{(k)}
\end{aligned}
\end{equation}
where:
\begin{equation}
    \begin{aligned}
        S^{(k+1)} &= H_{\mathbb{A},\mathbb{A}}^{(k+1)} + \rho_{\mathbb{A}}^{(k)} - H^{(k+1)}_{\mathbb{A}, \mathbb{I}} (P^{(k+1)})^{-1} H_{\mathbb{I},\mathbb{A}}^{(k+1)} \\
        B^{(k+1)} &= \begin{bmatrix} H_{\mathbb{A},\mathbb{A}}^{(k)} & H_{\mathbb{A}, \mathbb{I}}^{(k)} - H_{\mathbb{A},\mathbb{I}}^{(k+1)}(P^{(k+1)})^{-1} \rho_{\mathbb{I}}^{(k)} \end{bmatrix}
    \end{aligned}
\end{equation}
We now make two important observations. First, notice that because Proposition~\ref{prop:rho_dynamics} states that $\rho_{\mathbb{I}}^{(k)}$ is exponentially decreasing and $\|H^{(k)}\| \leq M$, $B^{(k+1)}$ is bounded in norm, uniformly over $k$.
\\
The second observation is that for some $c_0 > 0$, we have:
\begin{equation}
\label{eq:lambda_min_bound}
    \lambda_{\min}\left(S^{(k+1)}\right) > c_0 \beta^{k} \,.
\end{equation}
This is an consequence of the fact that, since $P^{(k+1)} \succ H_{\mathbb{I},\mathbb{I}}^{(k+1)}$:
\begin{equation}
\begin{aligned}
    H_{\mathbb{A},\mathbb{A}}^{(k+1)} - H^{(k+1)}_{\mathbb{A}, \mathbb{I}} (P^{(k+1)})^{-1} H_{\mathbb{I},\mathbb{A}}^{(k+1)} &\succ \\
    H_{\mathbb{A},\mathbb{A}}^{(k+1)} - H^{(k+1)}_{\mathbb{A}, \mathbb{I}} (H_{\mathbb{I}, \mathbb{I}}^{(k+1)})^{-1} H_{\mathbb{I},\mathbb{A}}^{(k+1)} &\succ 0
\end{aligned}
\end{equation}
and that according to Proposition~\ref{prop:rho_dynamics}:
\begin{equation}
    \lambda_{\min}\left(\rho_\mathbb{A}^{(k)}\right) \ge \beta^k \lambda_{\min}\left(\rho^{(0)}\right) \,.
\end{equation}
Equation~\eqref{eq:lambda_min_bound} then implies:
\begin{equation}
    \|\left(S^{(k+1)}\right)^{-1}\| = \frac{1}{\lambda_{\min}\left(S^{(k+1)}\right)} < \frac{1}{c_0} \left(\frac{1}{\beta}\right)^{k}
\end{equation}

Using these two observations together with Equation~\eqref{eq:solve_A}, we deduce that there exists $c_1 > 0$ such that:
\begin{equation}
\label{eq:e_A_norm_bound}
\|e_{\mathbb{A}}^{(k+1)}\| \leq c_1 \left(\frac{1}{\beta}\right)^{k} \|e^{(k)}\|
\end{equation}
Then from Equation~\eqref{eq:solve_I} we get:
\begin{equation}
\begin{aligned}
        \|e_{\mathbb{I}}^{(k+1)}\| \leq \frac{1}{m} \left( \|H_{\mathbb{I},\mathbb{A}}^{(k+1)}\| \|e_{\mathbb{A}}^{(k+1)}\|
        +  \|\rho_{\mathbb{I}}^{(k)}\|\|e_{\mathbb{I}}^{(k)}\|\right)
\end{aligned}
\end{equation}
where we used the fact that:
\begin{equation}
\begin{aligned}
    \|\left( P^{(k+1)} \right)^{-1}\| &= \frac{1}{\lambda_{\min}\left(H_{\mathbb{I}, \mathbb{I}}^{(k+1)} + \rho_{\mathbb{I}}^{(k)}\right)}\\
    &\leq \frac{1}{\lambda_{\min}\left( H_{\mathbb{I}, \mathbb{I}}^{(k+1)} \right)} \\
    &\leq \frac{1}{m} \,.
\end{aligned}
\end{equation}
Therefore, using $\|H^{(k+1)}_{\mathbb{I}, \mathbb{A}} \| \leq \|H^{(k+1)}\| \leq M$, $\|\rho_{\mathbb{I}}^{(k)}\| \leq \gamma^k \|\rho^0\|$ and the bound for $\|e_{\mathbb{A}}^{(k+1)}\|$ from Equation~\eqref{eq:e_A_norm_bound}, we get:
\begin{equation}
\label{eq:e_I_norm_bound}
\begin{aligned}
    \|e_{\mathbb{I}}^{(k+1)}\| \leq \frac{M}{m} c_1 \left(\frac{1}{\beta}\right)^k \|e^{(k)}\| + \frac{\|\rho^0\|}{m} \gamma^k \|e_{\mathbb{I}}^{(k)}\|
\end{aligned}
\end{equation}
Defining $\bar{\gamma} = \max\{\gamma, \frac{1}{\beta}\} \in (0,1)$, Inequalities~\eqref{eq:e_A_norm_bound} and ~\eqref{eq:e_I_norm_bound} together imply that there exists $c_2>0$ such that:
\begin{equation}
    \label{eq:error_norm_bound}
    \|e^{(k+1)}\| \leq c_2 \bar{\gamma}^k \|e^{(k)}\|
\end{equation}
Unrolling~\eqref{eq:error_norm_bound}, we get:
\begin{equation}
    \begin{aligned}
        \|e^{(k+1)}\| &\leq c_2 \bar{\gamma}^k \|e^{(k)}\| \\
        &\leq c_2^2 \bar{\gamma}^{k} \bar{\gamma}^{k-1} \|e^{(k-1)}\| \\
        &\vdots \\
        &\leq c_2^{k+1} \bar{\gamma}^{\left(\sum_{s=1}^k s\right)} \|e^{(0)}\| \\
        &= c_2^{k+1} \bar{\gamma}^{\frac{(k+1)k}{2}} \|e^{(0)}\| \\
        &\leq\exp \left(\frac{\log \bar{\gamma}}{2}(k+1) k + |\log c_2|(k+1)\right) \|e^{(0)}\| \\
        &= \exp \left( -\Theta\left(k^2\right)\right) \|e^{(0)}\|
    \end{aligned}
\end{equation}
where we used the fact that $\bar{\gamma} \in (0, 1)$. 
The $e^{-\Theta(k^2)}$ bound on $\|w^{(k)} - w^\star\|$ then extends directly to the KKT residuals $\|r_p^{(k)}\|_{\infty}$ and $\|r_d^{(k)}\|_{\infty}$.
\end{proof}

Before we proceed, a few remarks are in order. First, in Theorem~\ref{thm:superlinear_convergence} we considered box-constrained problems where $G=I$. However we have no reason to believe that a superlinear rate could not be established for general $G$. In fact, in our set of benchmark problems only Aircraft is box-constrained, and fast convergence was observed in all cases. 

Another assumption that needs to be addressed is $\mathcal{C}^2$ continuity of $f$ and positive-definiteness of $\nabla^2f$ about $w^{\star}$. While this clearly holds for condensed problems~\eqref{eq:dense_qp}, it does not for the sparse formulation~\eqref{eq:sparse_qp}, due to the presence of the non-smooth term $\mathcal{I}_{\mathcal{X}}$. However, one can prove that when running ADMM on the two problems with a compatible initialization, the $y^{(k)}$ iterates are exactly the same, while $z^{(k)}$ iterates coincide up to a constant 
arising from the linear-in-$x_0$ bound offset in the condensed formulation. Since $\rho^{(k)}$ is determined entirely 
by $y^{(k)}$, $y^{(k)} - y^{(k-1)}$ and $z^{(k)} - z^{(k-1)}$, it too coincides 
across the two formulations. As the primal and dual residuals depend on the same 
quantities, the superlinear convergence rate established above applies equally to 
Problem~\eqref{eq:sparse_qp}.

Our final remark is concerned with the fact that in Theorem~\ref{thm:superlinear_convergence} we assumed $\underline{\rho} = 0$ and $\overline{\rho} = + \infty$, essentially removing the effect of clipping from~\eqref{eq:update_rule}. However in a practical implementation clipping is actually necessary, as it prevents numerical ill-conditioning. Therefore, in order to preserve the benefits of clipping while enforcing the $\rho^{(k)}$-dynamics derived in Proposition~\ref{prop:rho_dynamics}, one could introduce a time-varying law for $\underline{\rho}$ and $\overline{\rho}$, allowing them to evolve as:
\begin{equation}
    \begin{aligned}
        \underline{\rho}^{(k+1)} &= \lambda_1 \underline{\rho} ^{(k)} \\
        \overline{\rho}^{(k+1)} &= \lambda_2 \overline{\rho} ^{(k)} 
    \end{aligned}
\end{equation}
where $\lambda_1 \in (0,1)$ and $\lambda_2 > 1$ should be chosen close enough to 1 to preserve numerical stability. However, we found this modification to have no practical impact on the convergence speed, and for simplicity kept $\underline{\rho}$ and $\overline{\rho}$ fixed.

\section{High Performance Implementation} \label{sec:implementation}

\begin{figure}[t] 
    \centering 
    \includegraphics[width=\linewidth]{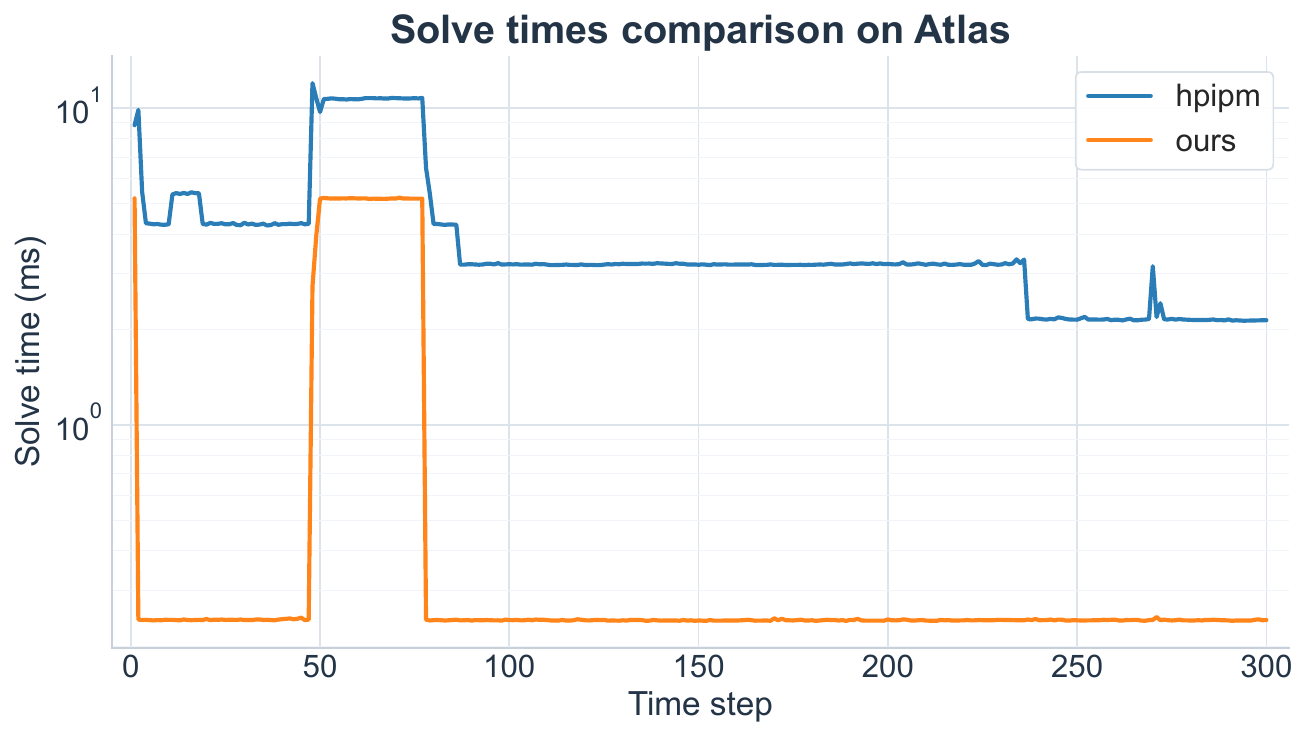} 
    \caption{Solve times comparison between the proposed solver and HPIPM on the Atlas problem.}
\label{fig:atlas_solve_times}
\end{figure}

In this section, we evaluate two high-performance implementations of the proposed ADMM-based MPC solver, corresponding to the sparse formulation~\eqref{eq:sparse_qp} and the condensed formulation~\eqref{eq:dense_qp}. Our goal is to assess both latency and warm-start performance, and to compare each implementation against state-of-the-art solvers designed for the same problem class. Both implementations rely on dense linear algebra kernels specialized for MPC, in the same spirit as BLASFEO~\cite{blasfeo}, while the sparse solver further exploits the OCP structure of the ADMM \(w\)-update.

To understand the runtime behavior of the proposed method, it is useful to recall the cost of updating the penalty parameter \(\rho^{(k)}\). For condensed MPC problems~\eqref{eq:dense_qp}, the ADMM \(w\)-update in~\eqref{eq:admm_iters} can be written as
\begin{equation}
\begin{aligned}
w^{(k+1)} &= - \left(H + G^{\top} \rho^{(k)} G \right)^{-1} \xi^{(k)} \\
\xi^{(k)} &= g + G^{\top} \left(y^{(k)} - \rho^{(k)} z^{(k)}\right) \,.
\end{aligned}
\end{equation}
The main cost of changing \(\rho^{(k)}\) is therefore not the evaluation of the update rule~\eqref{eq:update_rule} itself, but the need to refactorize the matrix \(H + G^{\top} \rho^{(k)} G\). In MPC, however, warm starts are often highly effective, and the solver frequently converges before the first update of \(\rho^{(k)}\) is triggered. In those cases, no factorization is required at solve time, which substantially reduces latency. The same consideration applies, mutatis mutandis, to the sparse formulation~\eqref{eq:sparse_qp}.

\begin{table}
\centering
\begin{tabular}{llcc}
\toprule
Problem & Solve time & Ours & HPIPM \\
\midrule
\multirow{2}{*}{Atlas} 
    & Avg (ms) & \textbf{0.74} & 4.02 \\
    & Max (ms) & \textbf{5.22} & 12.0 \\
\midrule
\multirow{2}{*}{Quadruped} 
    & Avg (ms) & \textbf{0.27} & 8.42 \\
    & Max (ms) & \textbf{4.86} & 45.68 \\
\bottomrule
\end{tabular}
\caption{Solve times comparison between the proposed sparse solver and HPIPM on Atlas and Quadruped.}
\label{tab:riccati_based_solve_times}
\end{table}

In all comparisons we set the tolerance for termination conditions to $\varepsilon = 10^{-6}$ for every solver. We compare the first version against HPIPM, which also exploits OCP structure, on the Atlas and Quadruped benchmarks. These problems are a good match for the sparse formulation, because they involve system models with a large number of control variables. Average and worst case solve times are reported in Table~\ref{tab:riccati_based_solve_times}, and we compare solve times over the whole Atlas simulation in Figure~\ref{fig:atlas_solve_times}. In both cases our solver outperforms HPIPM, in particular in terms of average solve times, highlighting the better warm-starting capabilities of our approach compared to an Interior Point method.
\begin{table}
\centering
\begin{tabular}{llcccc}
\toprule
Problem & Solve time & Ours & OSQP & qpOASES & QPALM \\
\midrule
\multirow{2}{*}{Quadrotor} 
    & Avg (ms) & \textbf{0.19} & 2.43 & \textbf{0.19} & 1.85 \\
    & Max (ms) & \textbf{2.24} & 31.78 & 12.99 & 7.42 \\
\midrule
\multirow{2}{*}{C.o.m} 
    & Avg (ms) & 2.91 & 10.0 & \textbf{1.96} & 3.2 \\
    & Max (ms) & 18.13 & 52.89 & 33.26 & \textbf{8.8} \\
\midrule
\multirow{2}{*}{Aircraft} 
    & Avg (ms) & \textbf{0.028} & - & 0.034 & - \\
    & Max (ms) & \textbf{0.087} & - & 0.34  & - \\
\bottomrule
\end{tabular}
\caption{Solve times comparison between the proposed dense solver and other state of the art solvers on Quadrotor, Chain of masses, Aircraft. On Aircraft, QPALM and OSQP fail to solve the problem within the maximum iteration budget.}
\end{table}
The second solver is compared against OSQP, qpOASES and QPALM \cite{Stellato2020, Ferreau2014, qpalm} on the remaining problems. The results again showcase how the proposed algorithm is competitive both in terms of latency and throughput. Relative to QPALM in particular, performance is slightly worse on the C.o.m. problem, and we suspect that this is due to the fact that in this problem the constraint matrix $G$ is very high dimensional and sparse, thus favoring solvers based on sparse linear algebra methods like QPALM.

\section{Conclusion} \label{sec:conclusions}
In this paper we introduced an adaptive multi-parameter ADMM algorithm that achieves significantly faster convergence speed compared to other ADMM variants, and proved its local superlinear convergence. 
We also compared the proposed algorithm against state-of-the-art MPC solvers, demonstrating compelling performance both in terms of latency and throughput, making it well-suited to latency-critical embedded MPC applications.

In future research we plan to extend our convergence result to non-box-constrained 
problems, and to analyze the transient behavior of our algorithm with the aim of 
establishing a linear convergence phase, finite constraint identification, and local 
superlinear convergence once the active set has been identified. We also plan to 
release the high-performance implementation as an open-source package for MPC.

\bibliographystyle{unsrt}

\end{document}